\documentclass{article}%
\usepackage{amsmath}
\usepackage{amsfonts}
\usepackage{amssymb}
\usepackage{graphicx}
\usepackage{tikz}
\usepackage[ruled,vlined,linesnumbered]{algorithm2e}%
\providecommand{\U}[1]{\protect\rule{.1in}{.1in}}
\newtheorem{theorem}{Theorem}

\newtheorem{conjecture}[theorem]{Conjecture}
\newtheorem{corollary}[theorem]{Corollary}

\newtheorem{example}[theorem]{Example}

\newtheorem{lemma}[theorem]{Lemma}

\newenvironment{proof}[1][Proof]{\noindent\textbf{#1.} }{\ \rule{0.5em}{0.5em}}
\begin{document}

\author{Vadim E. Levit\\Department of Computer Science\\Ariel University, ISRAEL\\levitv@ariel.ac.il
\and David Tankus\\Department of Software Engineering\\Sami Shamoon College of Engineering, ISRAEL\\davidt@sce.ac.il}
\title{Recognizing generating subgraphs revisited}
\date{}
\maketitle

\begin{abstract}
A graph $G$ is \textit{well-covered} if all its maximal independent sets are
of the same cardinality. Assume that a weight function $w$ is defined on its
vertices. Then $G$ is $w$\textit{-well-covered} if all maximal independent
sets are of the same weight. For every graph $G$, the set of weight functions
$w$ such that $G$ is $w$-well-covered is a \textit{vector space}, denoted as
$WCW(G)$. Deciding whether an input graph $G$ is well-covered is
\textbf{co-NP}-complete. Therefore, finding $WCW(G)$ is \textbf{co-NP}-hard.

A \textit{generating subgraph} of a graph $G$ is an induced complete bipartite
subgraph $B$ of $G$ on vertex sets of bipartition $B_{X}$ and $B_{Y}$, such
that each of $S \cup B_{X}$ and $S \cup B_{Y}$ is a maximal independent set of
$G$, for some independent set $S$. If $B$ is generating, then $w(B_{X}%
)=w(B_{Y})$ for every weight function $w \in WCW(G)$. Therefore, generating
subgraphs play an important role in finding $WCW(G)$.

The decision problem whether a subgraph of an input graph is generating is
known to be \textbf{NP-}complete. In this article we prove \textbf{NP-}%
completeness of the problem for graphs without cycles of length 3 and 5, and
for bipartite graphs with girth at least 6. On the other and, we supply
polynomial algorithms for recognizing generating subgraphs and finding
$WCW(G)$, when the input graph is bipartite without cycles of length 6. We
also present a polynomial algorithm which finds $WCW(G)$ when $G$ does not
contain cycles of lengths 3, 4, 5, and 7.

\end{abstract}

\section{Introduction}

\subsection{Basic Definitions and Notation}

Throughout this paper $G$ is a simple (i.e., a finite, undirected, loopless
and without multiple edges) graph with vertex set $V(G)$ and edge set $E(G)$.

Cycles of $k$ vertices are denoted by $C_{k}$. When we say that $G$ does not
contain $C_{k}$ for some $k \geq3$, we mean that $G$ does not admit subgraphs
isomorphic to $C_{k}$. It is important to mention that these subgraphs are not
necessarily induced.

Let $u$ and $v$ be two vertices in $G$. The \textit{distance} between $u$ and
$v$, denoted as $d(u,v)$, is the length of a shortest path between $u$ and
$v$, where the length of a path is the number of its edges. If $S$ is a
non-empty set of vertices, then the \textit{distance} between $u$ and $S$,
denoted as $d(u,S)$, is defined by
\[
d(u,S)=\min\{d(u,s):s\in S\}.
\]
For every positive integer $i$, denote
\[
N_{i}(S)=\{x\in V\left(  G\right)  :d(x,S)=i\},
\]
and
\[
N_{i}\left[  S\right]  =\{x\in V\left(  G\right)  :d(x,S)\leq i\}.
\]

We abbreviate $N_{1}(S)$ and $N_{1}\left[  S\right]  $ to be $N(S)$ and
$N\left[  S\right]  $, respectively. If $S$ contains a single vertex, $v$,
then we abbreviate $N_{i}(\{v\})$, $N_{i}\left[  \{v\}\right]  $, $N(\{v\})$,
and $N\left[  \{v\}\right]  $ to be $N_{i}(v)$, $N_{i}\left[  v\right]  $,
$N(v)$, and $N\left[  v\right]  $, respectively. We denote by $G[S]$ the
subgraph of $G$ induced by $S$. For every two sets, $S$ and $T$, of vertices
of $G$, we say that $S$ \textit{dominates} $T$ if $T\subseteq N\left[
S\right]  $.

\subsection{Well-Covered Graphs}

Let $G$ be a graph. A set of vertices $S$ is \textit{independent} if its
elements are pairwise nonadjacent. An independent set of vertices is
\textit{maximal} if it is not a subset of another independent set. An
independent set of vertices is \textit{maximum} if the graph does not contain
an independent set of a higher cardinality.

The graph $G$ is \textit{well-covered} if every maximal independent set is
maximum \cite{plummer:definition}. Assume that a weight function $w:V\left(
G\right)  \longrightarrow\mathbb{R}$ is defined on the vertices of $G$. For
every set $S\subseteq V\left(  G\right)  $, define $w(S)=%
{\displaystyle\sum\limits_{s\in S}}
w(s)$. Then $G$ is $w$\textit{-well-covered} if all maximal independent sets
of $G$ are of the same weight.

The problem of finding a maximum independent set in an input graph is
\textbf{NP-}hard. However, if the input is restricted to well-covered graphs,
then a maximum independent set can be found polynomially using the
\textit{greedy algorithm}. Similarly, if a weight function $w:V\left(
G\right)  \longrightarrow\mathbb{R}$ is defined on the vertices of $G$, and
$G$ is $w$-well-covered, then finding a maximum weight independent set is a
polynomial problem.

The recognition of well-covered graphs is known to be \textbf{co-NP}-complete.
This is proved independently in \cite{cs:note} and \cite{sknryn:compwc}. In
\cite{cst:structures} it is proven that the problem remains \textbf{co-NP}%
-complete even when the input is restricted to $K_{1,4}$-free graphs. However,
the problem is polynomially solvable for $K_{1,3}$-free graphs
\cite{tata:wck13f,tata:wck13fn}, for bipartite graphs
\cite{ravindra:well-covered}, for graphs with girth $5$ at least
\cite{fhn:wcg5}, for graphs with a bounded maximal degree \cite{cer:degree},
for chordal graphs \cite{ptv:chordal}, and for graphs without cycles of
lengths $4$ and $5$ \cite{fhn:wc45}. It should be emphasized that the
forbidden cycles are not necessarily induced.

For every graph $G$, the set of weight functions $w$ for which $G$ is
$w$-well-covered is a vector space \cite{cer:degree}. That vector space is
denoted $WCW(G)$ \cite{bnz:wcc4}. Since recognizing well-covered graphs is
\textbf{co-NP}-complete, finding the vector space $WCW(G)$ of an input graph
$G$ is \textbf{co-NP}-hard. However, finding $WCW(G)$ can be done in
polynomial time when the input is restricted to $K_{1,3}$-free graphs
\cite{lt:wwcclaw}, to graphs with a bounded maximal degree \cite{cer:degree},
to graphs without cycles of lengths $4$, $5$ and $6$ \cite{lt:wwc456}, and to
chordal graphs \cite{bn:wcchordal}.

\subsection{Generating Subgraphs and Relating Edges}

Further we make use of the following notions, which have been introduced in
\cite{lt:wc4567}. Let $B$ be an induced complete bipartite subgraph of $G$ on
vertex sets of bipartition $B_{X}$ and $B_{Y}$. Assume that there exists an
independent set $S$ such that each of $S\cup B_{X}$ and $S\cup B_{Y}$ is a
maximal independent set of $G$. Then $B$ is a \textit{generating} subgraph of
$G$, and the set $S$ is a \textit{witness} that $B$ is generating. We observe
that every weight function $w$ such that $G$ is $w$-well-covered must satisfy
the restriction $w(B_{X})=w(B_{Y})$.

In the restricted case that the generating subgraph $B$ is isomorphic to
$K_{1,1}$, call its vertices $x$ and $y$. In that case $xy$ is a
\textit{relating} edge, and $w(x)=w(y)$ for every weight function $w$ such
that $G$ is $w$-well-covered.

Recognizing relating edges is known to be \textbf{NP-}complete \cite{bnz:wcc4}%
, and it remains \textbf{NP-}complete even when the input is restricted to
graphs without cycles of lengths $4$ and $5$ \cite{lt:relatedc4}, and to
bipartite graphs \cite{lt:complexity}. Therefore, recognizing generating
subgraphs is also \textbf{NP-}complete for these restricted cases. However,
recognizing relating edges can be done in polynomial time if the input is
restricted to graphs without cycles of lengths $4$ and $6$ \cite{lt:relatedc4}%
, to graphs without cycles of lengths $5$ and $6$ \cite{lt:wwc456}, and to
graphs with a bounded maximal degree \cite{lt:complexity}.

It is also known that recognizing generating subgraphs is \textbf{NP-}complete
for graphs with girth at least $6$ \cite{lt:complexity}, and for $K_{1,4}%
$-free graphs \cite{lt:complexity}. However, the problem is a polynomial
solvable when the input is restricted to graphs without cycles of lengths $4$,
$6$ and $7$ \cite{lt:wc4567}, to graphs without cycles of lengths $4$, $5$ and
$6$ \cite{lt:wwc456}, to graphs without cycles of lengths $5$, $6$ and $7$
\cite{lt:wwc456}, to claw-free graphs \cite{tata:wck13fn}, and to graphs with
a bounded maximal degree \cite{lt:complexity}.

\subsection{Main Results}

This paper is a continuation of the research performed in \cite{lt:complexity}.

Two restricted cases of the well-known \textbf{SAT} problem are presented in
\cite{lt:complexity}, and proved to be \textbf{NP-}complete. In Section 2 we
use these results to prove that another restricted case of the \textbf{SAT}
problem, called \textbf{DMSAT}, is \textbf{NP-}complete as well.

In Section 3 we prove \textbf{NP-}completeness of two restricted cases of the
recognizing generating subgraphs problem. In Subsection 3.1 we consider the
case in which the input graph $G$ does not contain cycles of lengths $3$ and
$5$, and $B$ is $K_{1,2}$. We use the main result of Section 2 for this proof.
In Subsection 3.2 we deal with bipartite graphs with girth at least $6$.

In Section 4 we present a polynomial algorithm for recognizing generating
subgraphs of bipartite graphs without cycles of length 6. For this family of
graphs we also supply a polynomial algorithm which finds $WCW(G)$.

Section 5 contains a polynomial algorithm which finds $WCW(G)$ for graphs
without cycles of lengths 3, 4, 5, and 7. Especially, the algorithm works for
bipartite graphs with girth at least 6.

The following open question is presented in \cite{lt:complexity}. Does there
exist a family of graphs for which recognizing generating subgraphs is a
polynomial task, but finding $WCW(G)$ is \textbf{co}-\textbf{NP-}hard.
Although we still do not know the answer for this question, Subsection 3.2 and
Section 5 together give a first known example for the opposite case: a family
of graphs for which recognizing generating subgraphs is an \textbf{NP-}%
complete problem, but finding the vector space $WCW(G)$ is a polynomial task.

\section{Binary Variables}

A \textit{binary variable} is a variable whose value is either $0$ or $1$. If
$x$ is a binary variable, than its \textit{negation} is denoted by
$\overline{x}$. Each of $x$ and $\overline{x}$ is called a \textit{literal}.
Let $X=\{x_{1},...,x_{n}\}$ be a set of binary variables. A \textit{clause}
$c$ over $X$ is a set of literals belonging to $\{x_{1},\overline{x_{1}%
},...,x_{n},\overline{x_{n}}\}$ such that $c$ does not contain both a variable
and its negation. A \textit{truth assignment} is a function
\[
\Phi:\{x_{1},\overline{x_{1}},...,x_{n},\overline{x_{n}}\}\longrightarrow
\{0,1\}
\]
such that
\[
\Phi(\overline{x_{i}})=1-\Phi(x_{i})\text{ for each }1\leq i\leq n.
\]
A truth assignment $\Phi$ \textit{satisfies} a clause $c$ if $c$ contains at
least one literal $l$ such that $\Phi(l)=1$.

In \cite{lt:complexity} the following problems about binary variables are
presented.\newline\textbf{MONOTONE SAT }problem \cite{Gold1978},
\cite{lt:complexity}:\newline\textit{Input}: A set $X$ of binary variables and
two sets, $C_{1}$ and $C_{2}$, of clauses over $X$, such that all literals of
the clauses belonging to $C_{1}$ are variables, and all literals of clauses
belonging to $C_{2}$ are negations of variables.\newline\textit{Question}: Is
there a truth assignment for $X$, which satisfies all clauses of $C=C_{1}\cup
C_{2}$?

\begin{theorem}
\cite{Gold1978}\cite{lt:complexity}\label{monotonesat} The \textbf{MONOTONE
SAT} problem is \textbf{NP-}complete.
\end{theorem}

\textbf{DSAT }problem \cite{lt:complexity}:\newline\textit{Input}: A set $X$
of binary variables and a set $C$ of clauses over $X$ such that the following holds:

\begin{itemize}
\item Every clause contains $2$ or $3$ literals.

\item Every two clauses have at most one literal in common.

\item If two clauses, $c_{1}$ and $c_{2}$, have a common literal $l_{1}$, then
there does not exist a literal $l_{2}$ such that $c_{1}$ contains $l_{2}$ and
$c_{2}$ contains $\overline{l_{2}}$.
\end{itemize}

\textit{Question}: Is there a truth assignment for $X$ which satisfies all
clauses of $C$?

\begin{theorem}
\cite{lt:complexity}\label{dsatnpc} The \textbf{DSAT} problem is \textbf{NP-}complete.
\end{theorem}

In this paper, we define the following problem.\newline\textbf{DMSAT
}problem:\newline\textit{Input}: A set $X$ of binary variables and two sets,
$C_{1}$ and $C_{2}$, of clauses over $X$, such that the following holds:

\begin{itemize}
\item All literals of the clauses belonging to $C_{1}$ are variables.

\item All literals of the clauses belonging to $C_{2}$ are negations of variables.

\item Every clause of $C_{1}$ contains $2$ or $3$ literals.

\item Every clause of $C_{2}$ contains $2$ literals.

\item Every two clauses of $C_{1}$ have at most one literal in common.

\item Every two clauses of $C_{2}$ are disjoint.
\end{itemize}

\textit{Question}: Is there a truth assignment for $X$ which satisfies all
clauses of $C = C_{1} \cup C_{2}$?

\begin{theorem}
\label{dusat} The \textbf{DMSAT} problem is \textbf{NP-}complete.
\end{theorem}

\begin{proof}
Obviously, the \textbf{DMSAT} problem is \textbf{NP}. We prove its
\textbf{NP-}completeness by showing a reduction from the \textbf{DSAT}
problem. Let
\[
I_{1}=(X=\{x_{1},...,x_{n}\},C=\{c_{1},...,c_{m}\})
\]
be an instance of the \textbf{DSAT} problem. Define $Z=\{x_{1},...,x_{n}%
,z_{1},...,z_{n}\}$, where $z_{1},...,z_{n}$ are new variables. For every
$1\leq j\leq m$, let $c_{j}^{\prime}$ be the clause obtained from $c_{j}$ by
replacing $\overline{x_{i}}$ with $z_{i}$ for each $1\leq i\leq n$. Let
$C^{\prime}=\{c_{1}^{\prime},...,c_{m}^{\prime}\}$. For each $1\leq i\leq n$
define two new clauses, $d_{i}=(x_{i},z_{i})$ and $e_{i}=(\overline{x_{i}%
},\overline{z_{i}})$. Let $D=\{d_{1},...,d_{n}\}$ and $E=\{e_{1},...,e_{n}\}$.
Define $I_{2}=(Z,C^{\prime}\cup D\cup E)$.

We show that $I_{2}$ is an instance of the \textbf{DMSAT} problem. Obviously,
all literals of $C^{\prime}\cup D$ are variables, and all literals of $E$ are
negations of variables. Moreover, every clause of $C^{\prime}$ contains $2$ or
$3$ literals, and every clause of $D\cup E$ contains $2$ literals. The fact
that there are no $2$ clauses of $C$ with $2$ common literals implies that
there are no 2 clauses of $C^{\prime}$ with $2$ common literals. A clause in
$C^{\prime}$ and a clause in $D$ can not have $2$ common literals, since a
clause in $C$ can not contain both a variable and its negation. A clause in
$C^{\prime}\cup D$ and a clause in $E$ do not have common literals, since all
literals of $C^{\prime}\cup D$ are variables and all literals of $E$ are
negations of variables. By definition of $E$, if $i\neq j$ then $e_{i}%
=(\overline{x_{i}},\overline{z_{i}})$ and $e_{j}=(\overline{x_{j}}%
,\overline{z_{j}})$ are disjoint. Hence, $I_{2}$ is an instance of the
\textbf{DMSAT} problem, see Example \ref{exdusat}. It remains to prove that
$I_{1}$ and $I_{2}$ are equivalent.

Assume that $I_{1}$ is a positive instance of the \textbf{DSAT} problem. There
exists a truth assignment
\[
\Phi_{1}:\{x_{1},\overline{x_{1}},...,x_{n},\overline{x_{n}}\}\longrightarrow
\{0,1\}
\]
which satisfies all clauses of $C$. Extract $\Phi_{1}$ to a truth assignment
\[
\Phi_{2}:\{x_{1},\overline{x_{1}},...,x_{n},\overline{x_{n}},z_{1}%
,\overline{z_{1}},...,z_{n},\overline{z_{n}}\}\longrightarrow\{0,1\}
\]
by defining $\Phi_{2}(z_{i})=1-\Phi_{1}(x_{i})$ for each $1\leq i\leq n$.
Clearly, $\Phi_{2}$ is a truth assignment which satisfies all clauses of
$C^{\prime}\cup D\cup E$. Hence, $I_{2}$ is a positive instance of the
\textbf{DMSAT} problem.

Assume $I_{2}$ is a positive instance of the \textbf{DMSAT} problem. There
exists a truth assignment
\[
\Phi_{2}:\{x_{1},\overline{x_{1}},...,x_{n},\overline{x_{n}},z_{1}%
,\overline{z_{1}},...,z_{n},\overline{z_{n}}\}\longrightarrow\{0,1\}
\]
that satisfies all clauses of $C^{\prime}\cup D\cup E$. For every $1\leq i\leq
n$ it holds that $\Phi_{2}(z_{i})=1-\Phi_{2}(x_{i})$, or otherwise one of
$d_{i}$ and $e_{i}$ is not satisfied. Therefore, $I_{1}$ is a positive
instance of the \textbf{DSAT} problem.
\end{proof}

\begin{example}
\label{exdusat} The following contains an instance of the \textbf{DSAT}
problem and an equivalent instance of the \textbf{DMSAT} problem.

$I_{1}=(X,C_{1})$, where $X=\{x_{1},...,x_{9}\}$ and $C_{1}=\{\{x_{1}%
,\overline{x_{2}},x_{3}\},\{x_{1},x_{6},x_{4}\},\newline\{x_{1},x_{7}%
,x_{5}\},\{\overline{x_{3}},\overline{x_{4}},x_{5}\},\{x_{2},\overline{x_{3}%
},x_{8}\},\{\overline{x_{1}},\overline{x_{4}},x_{9}\},\{\overline{x_{3}}%
,x_{6}\},\{x_{3},\overline{x_{6}}\},\{\overline{x_{3}},x_{7}\},\newline%
\{x_{3},\overline{x_{7}}\},\{x_{4},x_{8}\},\{\overline{x_{4}},\overline{x_{8}%
}\},\{x_{5},x_{9}\},\{\overline{x_{5}},\overline{x_{9}}\}\}$.

$I_{2}=(Z,C_{2})$, where $Z=\{x_{1},...,x_{9},z_{1},...,z_{9}\}$ and
$C_{2}=\{\{x_{1},z_{2},x_{3}\},\newline\{x_{1},x_{6},x_{4}\},\{x_{1}%
,x_{7},x_{5}\},\{z_{3},z_{4},x_{5}\},\{x_{2},z_{3},x_{8}\},\{z_{1},z_{4}%
,x_{9}\},\{z_{3},x_{6}\},\{x_{3},z_{6}\},\newline\{z_{3},x_{7}\},\{x_{3}%
,z_{7}\},\{x_{4},x_{8}\},\{z_{4},z_{8}\},\{x_{5},x_{9}\},\{z_{5}%
,z_{9}\},\{x_{1},z_{1}\},\{\overline{x_{1}},\overline{z_{1}}\},\{x_{2}%
,z_{2}\},\newline\{\overline{x_{2}},\overline{z_{2}}\},\{x_{3},z_{3}%
\},\{\overline{x_{3}},\overline{z_{3}}\},\{x_{4},z_{4}\},\{\overline{x_{4}%
},\overline{z_{4}}\},\{x_{5},z_{5}\},\{\overline{x_{5}},\overline{z_{5}%
}\},\{x_{6},z_{6}\},\{\overline{x_{6}},\overline{z_{6}}\},\newline%
\{x_{7},z_{7}\},\{\overline{x_{7}},\overline{z_{7}}\},\{x_{8},z_{8}%
\},\{\overline{x_{8}},\overline{z_{8}}\},\{x_{9},z_{9}\},\{\overline{x_{9}%
},\overline{z_{9}}\}\}$.
\end{example}

\section{\textbf{NP-}Complete Results for Recognizing Generating Subgraphs}

The $\mathbf{GS}$ problem is defined as follows. \newline\textit{Input:} A
graph $G$ and an induced subgraph $B$.\newline\textit{Question:} Is $B$ generating?

The $\mathbf{GS}$ problem is known to be \textbf{NP-}complete \cite{bnz:wcc4}.
In this section we prove that it remains \textbf{NP-}complete for two
restricted cases: bipartite graphs with girth at least $6$, and graphs without
cycles of lengths $3$ and $5$.

\subsection{Bipartite Graphs with Girth at Least 6}

\begin{theorem}
\cite{lt:complexity} \label{gsg6bnpc} The following problem is \textbf{NP-}%
complete.\newline\textit{Input}: A graph $G$ with girth at least $6$, and a
subgraph $B$ of $G$.\newline\textit{Question}: Is $B$ generating?
\end{theorem}

Theorem \ref{gsg6bnpc} is an instance of Theorem \ref{bg6npc}.

\begin{theorem}
\label{bg6npc} The following problem is \textbf{NP-}complete.\newline%
\textit{Input}: A bipartite graph $G$ with girth at least $6$, and a subgraph
$B$ of $G$.\newline\textit{Question}: Is $B$ generating?
\end{theorem}

\begin{proof}
The problem is obviously \textbf{NP}. We prove \textbf{NP-}completeness by
showing a reduction from the \textbf{DMSAT} problem. Let
\[
I_{1}=(X=\{x_{1},...,x_{n}\},C=C_{1}\cup C_{2})
\]
be an instance of the \textbf{DMSAT} problem, where $C_{1}=\{c_{1}%
,...,c_{m}\}$ and $C_{2}=\{c_{1}^{\prime},...,c_{m^{\prime}}^{\prime}\}$ are
sets of clauses. Every clause of $C_{1}$ contains 2 or $3$ variables, and
every clause of $C_{2}$ contains $2$ negations of variables. Every $2$ clauses
of $C_{1}$ have at most one literal in common, and every two clauses of
$C_{2}$ are disjoint. Define a graph $G$ as follows:
\begin{gather*}
V\left(  G\right)  =\{x\}\cup\{y_{j}:1\leq j\leq m\}\cup\{v_{j}:1\leq j\leq
m\}\cup\{v_{j}^{\prime}:1\leq j\leq m^{\prime}\}\cup\\
\{u_{i}:1\leq i\leq n\}\cup\{u_{i}^{\prime}:1\leq i\leq n\},
\end{gather*}%
\begin{gather*}
E\left(  G\right)  =\{xy_{j}:1\leq j\leq m\}\cup\{y_{j}v_{j}:1\leq j\leq
m\}\cup\{xv_{j}^{\prime}:1\leq j\leq m^{\prime}\}\cup\\
\{v_{j}u_{i}:x_{i}\ \text{appears\ in}\ c_{j}\}\cup\{v_{j}^{\prime}%
u_{i}^{\prime}:\overline{x_{i}}\ \text{appears\ in}\ c_{j}^{\prime}%
\}\cup\{u_{i}u_{i}^{\prime}:1\leq i\leq n\}\newline.
\end{gather*}
\newline Clearly, $G$ is bipartite, and the vertex sets of its bipartition
are
\[
\{u_{i}:1\leq i\leq n\}\cup\{y_{j}:1\leq j\leq m\}\cup\{v_{j}^{\prime}:1\leq
j\leq m^{\prime}\}
\]
and
\[
\{v_{j}:1\leq j\leq m\}\cup\{x\}\cup\{u_{i}^{\prime}:1\leq i\leq n\}.
\]

Since $C_{1}$ does not contain two clauses with common two literals, and the
clauses of $C_{2}$ are pairwise disjoint, $G$ does not contain cycles of
length $4$. Hence, its girth is at least $6$. Let $B = G[\{x\}\cup
\{y_{j}:1\leq j\leq m\}]$, and let $I_{2}=(G,B)$ be an instance of the
$\mathbf{GS}$ problem. It remains to prove that $I_{1}$ and $I_{2}$ are equivalent.

Assume that $I_{1}$ is a positive instance of the \textbf{DMSAT} problem. Let
\[
\Phi:\{x_{1},\overline{x_{1}},...,x_{n},\overline{x_{n}}\}\longrightarrow
\{0,1\}
\]
be a truth assignment which satisfies all clauses of $C$. Let
\[
S=\{u_{i}:\Phi(x_{i})=1\}\cup\{u_{i}^{\prime}:\Phi(x_{i})=0\}.
\]
Obviously, $S$ is independent. Since $\Phi$ satisfies all clauses of $C$,
every vertex of
\[
\{v_{j}:1\leq j\leq m\}\cup\{v_{j}^{\prime}:1\leq j\leq m^{\prime}\}
\]
is adjacent to a vertex of $S$. Hence, $S$ is a witness that $B$ is a
generating subgraph of $G$. Therefore, $I_{2}$ is positive.

On the other hand, assume that $I_{2}$ is a positive instance of the
$\mathbf{GS}$ problem. Let $S$ be a witness of $B$. Since $S$ is a maximal
independent set of
\[
\{u_{i}:1\leq i\leq n\}\cup\{u_{i}^{\prime}:1\leq i\leq n\},
\]
exactly one of $u_{i}$ and $u_{i}^{\prime}$ belongs to $S$, for every $1\leq
i\leq n$. Let
\[
\Phi:\{x_{1},\overline{x_{1}},...,x_{n},\overline{x_{n}}\}\longrightarrow
\{0,1\}
\]
be a truth assignment defined by: $\Phi(x_{i})=1\iff u_{i}\in S$. The fact
that $S$ dominates
\[
\{v_{j}:1\leq j\leq m\}\cup\{v_{j}^{\prime}:1\leq j\leq m^{\prime}\}
\]
implies that all clauses of $C$ are satisfied by $\Phi$. Therefore, $I_{1}$ is
a positive instance of the \textbf{DMSAT} problem.
\end{proof}

\begin{figure}[h]
\setlength{\unitlength}{1.0cm} \begin{picture}(20,12)\thicklines
\put(6,5){\circle*{0.1}}
\multiput(3,6.5)(0.75,0){9}{\circle*{0.1}}
\multiput(3,8)(0.75,0){9}{\circle*{0.1}}
\multiput(2.25,9.5)(0.75,0){5}{\circle*{0.1}}
\put(2.25,9.5){\circle*{0.25}}
\put(3.75,9.5){\circle*{0.25}}
\put(6.75,9.5){\circle*{0.25}}
\put(7.5,9.5){\circle*{0.25}}
\put(9.75,9.5){\circle*{0.25}}
\multiput(6.75,9.5)(0.75,0){5}{\circle*{0.1}}
\put(6,6.5){\circle*{0.1}}
\put(5.7,5){\makebox(0,0){$x$}}
\put(2.7,6.5){\makebox(0,0){$y_{1}$}}
\put(3.45,6.5){\makebox(0,0){$y_{2}$}}
\put(4.2,6.5){\makebox(0,0){$y_{3}$}}
\put(4.95,6.5){\makebox(0,0){$y_{4}$}}
\put(5.7,6.5){\makebox(0,0){$y_{5}$}}
\put(6.45,6.5){\makebox(0,0){$y_{6}$}}
\put(7.2,6.5){\makebox(0,0){$y_{7}$}}
\put(7.95,6.5){\makebox(0,0){$y_{8}$}}
\put(8.7,6.5){\makebox(0,0){$y_{9}$}}
\put(2.7,8){\makebox(0,0){$v_{1}$}}
\put(3.45,8){\makebox(0,0){$v_{2}$}}
\put(4.2,8){\makebox(0,0){$v_{3}$}}
\put(4.95,8){\makebox(0,0){$v_{4}$}}
\put(5.7,8){\makebox(0,0){$v_{5}$}}
\put(6.45,8){\makebox(0,0){$v_{6}$}}
\put(7.2,8){\makebox(0,0){$v_{7}$}}
\put(7.95,8){\makebox(0,0){$v_{8}$}}
\put(8.7,8){\makebox(0,0){$v_{9}$}}
\put(1.95,9.5){\makebox(0,0){$u_{1}$}}
\put(2.7,9.5){\makebox(0,0){$u_{2}$}}
\put(3.45,9.5){\makebox(0,0){$u_{3}$}}
\put(4.2,9.5){\makebox(0,0){$u_{4}$}}
\put(4.95,9.5){\makebox(0,0){$u_{5}$}}
\put(7.05,9.5){\makebox(0,0){$u_{6}$}}
\put(7.8,9.5){\makebox(0,0){$u_{7}$}}
\put(8.55,9.5){\makebox(0,0){$u_{8}$}}
\put(9.3,9.5){\makebox(0,0){$u_{9}$}}
\put(10.15,9.5){\makebox(0,0){$u_{10}$}}
\put(6,5){\line(-2,1){3}}
\put(6,5){\line(-3,2){2.25}}
\put(6,5){\line(-1,1){1.5}}
\put(6,5){\line(-1,2){0.75}}
\put(6,5){\line(0,1){1.5}}
\put(6,5){\line(1,2){0.75}}
\put(6,5){\line(1,1){1.5}}
\put(6,5){\line(3,2){2.25}}
\put(6,5){\line(2,1){3}}
\multiput(3,6.5)(0.75,0){9}{\line(0,1){1.5}}
\put(3,8){\line(-1,2){0.75}}
\put(3,8){\line(1,1){1.5}}
\put(3,8){\line(3,2){2.25}}
\put(3.75,8){\line(-1,1){1.5}}
\put(3.75,8){\line(-1,2){0.75}}
\put(4.5,8){\line(-1,1){1.5}}
\put(4.5,8){\line(-1,2){0.75}}
\put(4.5,8){\line(3,2){2.25}}
\put(5.25,8){\line(-1,1){1.5}}
\put(5.25,8){\line(-1,2){0.75}}
\put(6,8){\line(-2,1){3}}
\put(6,8){\line(1,1){1.5}}
\put(6,8){\line(2,1){3}}
\put(6.75,8){\line(-2,1){3}}
\put(6.75,8){\line(1,1){1.5}}
\put(6.75,8){\line(2,1){3}}
\put(7.5,8){\line(-3,2){2.25}}
\put(7.5,8){\line(-1,2){0.75}}
\put(8.25,8){\line(-1,2){0.75}}
\put(8.25,8){\line(0,1){1.5}}
\put(9,8){\line(0,1){1.5}}
\put(9,8){\line(1,2){0.75}}
\multiput(3,3.5)(1.5,0){5}{\circle*{0.1}}
\multiput(2.25,2)(0.75,0){5}{\circle*{0.1}}
\multiput(6.75,2)(0.75,0){5}{\circle*{0.1}}
\put(3,2){\circle*{0.25}}
\put(4.5,2){\circle*{0.25}}
\put(5.25,2){\circle*{0.25}}
\put(8.25,2){\circle*{0.25}}
\put(9,2){\circle*{0.25}}
\put(2.7,3.5){\makebox(0,0){$v'_{1}$}}
\put(4.2,3.5){\makebox(0,0){$v'_{2}$}}
\put(5.7,3.5){\makebox(0,0){$v'_{3}$}}
\put(7.75,3.5){\makebox(0,0){$v'_{4}$}}
\put(9.25,3.5){\makebox(0,0){$v'_{5}$}}
\put(1.95,2){\makebox(0,0){$u'_{1}$}}
\put(2.7,2){\makebox(0,0){$u'_{2}$}}
\put(3.45,2){\makebox(0,0){$u'_{3}$}}
\put(4.2,2){\makebox(0,0){$u'_{4}$}}
\put(4.95,2){\makebox(0,0){$u'_{5}$}}
\put(7.1,2){\makebox(0,0){$u'_{6}$}}
\put(7.85,2){\makebox(0,0){$u'_{7}$}}
\put(8.6,2){\makebox(0,0){$u'_{8}$}}
\put(9.35,2){\makebox(0,0){$u'_{9}$}}
\put(10.1,2){\makebox(0,0){$u'_{10}$}}
\put(6,5){\line(-2,-1){3}}
\put(6,5){\line(-1,-1){1.5}}
\put(6,5){\line(0,-1){1.5}}
\put(6,5){\line(1,-1){1.5}}
\put(6,5){\line(2,-1){3}}
\put(3,3.5){\line(0,-1){1.5}}
\put(3,3.5){\line(-1,-2){0.75}}
\put(4.5,3.5){\line(0,-1){1.5}}
\put(4.5,3.5){\line(-1,-2){0.75}}
\put(6,3.5){\line(1,-2){0.75}}
\put(6,3.5){\line(-1,-2){0.75}}
\put(7.5,3.5){\line(0,-1){1.5}}
\put(7.5,3.5){\line(1,-2){0.75}}
\put(9,3.5){\line(0,-1){1.5}}
\put(9,3.5){\line(1,-2){0.75}}
\put(2.25,2){\line(0,1){7.5}}
\put(9.75,2){\line(0,1){7.5}}
\put(2.25,5.8){\oval(1.5,9.4)[l]}
\put(2.25,9.5){\oval(1.5,2)[t]}
\put(2.25,2){\oval(1.5,1.8)[b]}
\put(2.73,5.8){\oval(3.35,9.7)[l]}
\put(2.4,9.55){\oval(2.7,2.2)[t]}
\put(2.4,2){\oval(2.7,2.1)[b]}
\put(3.2,5.8){\oval(5.2,10.2)[l]}
\put(2.55,9.5){\oval(3.9,2.8)[t]}
\put(2.55,2.7){\oval(3.9,4)[b]}
\put(3.75,5.8){\oval(7.2,10.7)[l]}
\put(2.7,9.55){\oval(5.1,3.2)[t]}
\put(2.7,2){\oval(5.1,3.1)[b]}
\put(9.75,5.8){\oval(1.5,9.4)[r]}
\put(9.75,9.5){\oval(1.5,2)[t]}
\put(9.75,2){\oval(1.5,1.8)[b]}
\put(9.27,5.8){\oval(3.35,9.7)[r]}
\put(9.6,9.55){\oval(2.7,2.2)[t]}
\put(9.6,2){\oval(2.7,2.1)[b]}
\put(8.8,5.8){\oval(5.2,10.2)[r]}
\put(9.45,9.5){\oval(3.9,2.8)[t]}
\put(9.45,2.7){\oval(3.9,4)[b]}
\put(8.25,5.8){\oval(7.2,10.7)[r]}
\put(9.3,9.55){\oval(5.1,3.2)[t]}
\put(9.3,2){\oval(5.1,3.1)[b]}
\end{picture}
\caption{An example of the reduction from the \textbf{DMSAT} problem to the
\textbf{GS} problem.}%
\label{Fig1}%
\end{figure}
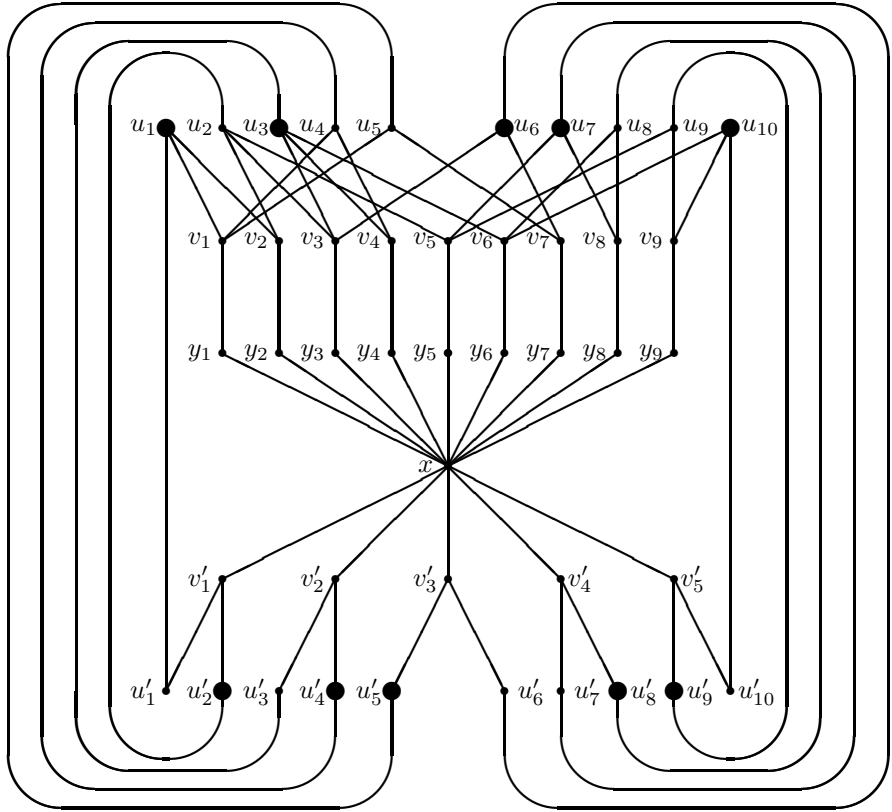

\begin{example}
\label{gsreduction} The following are an instance of the \textbf{DSAT}
problem, an equivalent instance of the \textbf{DMSAT} problem, and an
equivalent instance of the $\mathbf{GS}$ problem. $I_{1}=\{\{x_{1}%
,\overline{x_{3}},x_{5}\}, \{\overline{x_{1}},x_{3},\overline{x_{5}}\},
\{\overline{x_{1}},x_{7},x_{9}\}, \{x_{3},\overline{x_{7}},\overline{x_{9}%
}\}\}$ \newline$I_{2} = \{ \{x_{1},x_{4},x_{5}\}, \{x_{2},x_{3},x_{6}\},
\{x_{2},x_{7},x_{9}\}, \{x_{3},x_{8},x_{10}\}, \{x_{1},x_{2}\}, \{\overline
{x_{1}},\overline{x_{2}}\}, \newline\{x_{3},x_{4}\}, \{\overline{x_{3}%
},\overline{x_{4}}\}, \{x_{5},x_{6}\}, \{\overline{x_{5}},\overline{x_{6}}\},
\{x_{7},x_{8}\}, \{\overline{x_{7}},\overline{x_{8}}\}, \{x_{9},x_{10}\},
\{\overline{x_{9}},\overline{x_{10}}\}\}$ \newline$I_{3}$=$(G,B)$, where $G$
is the graph shown in Figure \ref{Fig1}, $B$=$G[\{x\}\cup\{y_{j}:1\leq
j\leq9\}]$.

The instance $I_{2}$ is positive because $x_{1}=x_{3}=x_{6}=x_{7}=x_{10}=1$,
$x_{2}=x_{4}=x_{5}=x_{8}=x_{9}=0$ is a satisfying assignment. The
corresponding witness that $I_{3}$ is positive is the set $\{u_{1}%
,u_{2}^{\prime},u_{3},u_{4}^{\prime},u_{5}^{\prime},u_{6},u_{7},u_{8}^{\prime
},u_{9}^{\prime},u_{10}\}$.
\end{example}

\subsection{Graphs Without Cycles of Lengths 3 and 5}

\begin{theorem}
\label{rebnpc} \cite{lt:complexity} The $\mathbf{GS}$ problem is
\textbf{NP-}complete even in the restricted case that $G$ is bipartite, and
$B$ is $K_{1,1}$.
\end{theorem}

Theorem \ref{k12c35npc} is the main result of this subsection.

\begin{theorem}
\label{k12c35npc} The $\mathbf{GS}$ problem is \textbf{NP-}complete even in
the restricted case that $G$ does not contain cycles of lengths $3$ and $5$,
and $B$ is $K_{1,2}$.
\end{theorem}

\begin{proof}
The problem is obviously in \textbf{NP}. We prove \textbf{NP-}completeness by
showing a reduction from the \textbf{MONOTONE SAT} problem. Let
\[
I_{1}=(X=\{x_{1},...,x_{n}\},C=C_{1}\cup C_{2})
\]
be an instance of the \textbf{MONOTONE SAT} problem, where $X=\{x_{1}%
,...,x_{n}\}$ is a set of $0-1$ variables, $C_{1}=\{c_{1},...,c_{m}\}$ is a
set of clauses over the literals $\{x_{1},...,x_{n}\}$, and $C_{2}%
=\{c_{1}^{\prime},...,c_{m^{\prime}}^{\prime}\}$ is a set of clauses over the
literals $\{\overline{x_{1}},...,\overline{x_{n}}\}$.

Define a graph $G$ as follows:
\begin{gather*}
V\left(  G\right)  =\{z,y_{1},y_{2}\}\cup\{v_{j}:1\leq j\leq m\}\cup
\{v_{j}^{\prime}:1\leq j\leq m^{\prime}\}\cup\\
\{u_{i}:1\leq i\leq n\}\cup\{u_{i}^{\prime}:1\leq i\leq n\},
\end{gather*}%
\begin{gather*}
E\left(  G\right)  =\{zy_{1},zy_{2}\}\cup\{y_{1}v_{j}:1\leq j\leq
m\}\cup\{y_{2}v_{j}^{\prime}:1\leq j\leq m^{\prime}\}\cup\\
\{v_{j}u_{i}:x_{i}\ \text{appears\ in}\ c_{j}\}\cup\{v_{j}^{\prime}%
u_{i}^{\prime}:\overline{x_{i}}\ \text{appears\ in}\ c_{j}^{\prime}%
\}\cup\{u_{i}u_{i}^{\prime}:1\leq i\leq n\}\newline.
\end{gather*}
\newline The next step is proving that $G$ does not contain cycles of lengths
$3$ and $5$. If the edges $\{u_{i}u_{i}^{\prime}:1\leq i\leq n\}$ are deleted
from $G$, then the graph becomes bipartite, with the vertex sets of
bipartition:
\[
\{u_{i}:1\leq i\leq n\}\cup\{u_{i}^{\prime}:1\leq i\leq n\}\cup\{y_{1}%
,y_{2}\}.
\]
and
\[
\{v_{j}:1\leq j\leq m\}\cup\{v_{j}^{\prime}:1\leq j\leq m^{\prime}%
\}\cup\{z\}.
\]
Hence, it is enough to prove that an edge $u_{i}u_{i}^{\prime}$ is not a part
of a cycle of length $3$ or $5$. The vertices $u_{i}$ and $u_{i}^{\prime}$
have no common neighbors. Therefore, $u_{i}u_{i}^{\prime}$ is not a part of a
triangle. Let $v_{j}\in N(u_{i})\setminus\{u_{i}^{\prime}\}$ and
$v_{j^{\prime}}^{\prime}\in N(u_{i}^{\prime})\setminus\{u_{i}\}$. Since
$v_{j}$ and $v_{j^{\prime}}^{\prime}$ have no common neighbors, $u_{i}%
u_{i}^{\prime}$ is not a part of a cycle of length $5$.

Let $B = G[\{z,y_{1},y_{2}\}]$, and let $I_{2} = (G,B)$ be an instance of the
\textbf{GS} problem. It remains to prove that $I_{1}$ and $I_{2}$ are equivalent.

Assume that $I_{1}$ is a positive instance of the \textbf{MONOTONE SAT}
problem. Let
\[
\Phi:\{x_{1},\overline{x_{1}},...,x_{n},\overline{x_{n}}\}\longrightarrow
\{0,1\}
\]
be a truth assignment which satisfies all clauses of $C$. Let
\[
S=\{u_{i}:\Phi(x_{i})=1\}\cup\{u_{i}^{\prime}:\Phi(x_{i})=0\}.
\]
Obviously, $S$ is independent. Since $\Phi$ satisfies all clauses of $C$,
every vertex of
\[
\{v_{j}:1\leq j\leq m\}\cup\{v_{j}^{\prime}:1\leq j\leq m^{\prime}\}
\]
is adjacent to a vertex of $S$. Hence, $S$ is a witness that $B$ is a
generating subgraph of $G$. Therefore, $I_{2}$ is positive.

On the other hand, assume that $I_{2}$ is a positive instance of the
\textbf{GS} problem. Let $S$ be a witness of $B$. Since $S$ is a maximal
independent set of
\[
\{u_{i}:1\leq i\leq n\}\cup\{u_{i}^{\prime}:1\leq i\leq n\},
\]
exactly one of $u_{i}$ and $u_{i}^{\prime}$ belongs to $S$, for every $1\leq
i\leq n$. Let
\[
\Phi:\{x_{1},\overline{x_{1}},...,x_{n},\overline{x_{n}}\}\longrightarrow
\{0,1\}
\]
be a truth assignment defined by: $\Phi(x_{i})=1\iff u_{i}\in S$. The fact
that $S$ dominates
\[
\{v_{j}:1\leq j\leq m\}\cup\{v_{j}^{\prime}:1\leq j\leq m^{\prime}\}
\]
implies that all clauses of $C$ are satisfied by $\Phi$. Therefore, $I_{1}$ is
a positive instance of the \textbf{MONOTONE SAT} problem.
\end{proof}

\begin{corollary}
\label{kpqc35npc} Let $p \geq1$ and $q \geq2$. The $\mathbf{GS}$ problem is
\textbf{NP-}complete even in the restricted case that $G$ does not contain
cycles of lengths $3$ and $5$, and $B$ is $K_{p,q}$.
\end{corollary}

\begin{proof}
We prove \textbf{NP-}completeness by showing a reduction from the
\textbf{MONOTONE SAT} problem. Let $I$ be an instance of the \textbf{MONOTONE
SAT} problem. Let $G$ be the graph constructed in the proof of Theorem
\ref{k12c35npc}, and contains the vertices ${y_{1},y_{2},z}$. Define $z_{1} =
z$, and let $H$ be the graph obtained from $G$ by adding vertices
$y_{3},...,y_{p},z_{2},...,z_{q}$ and edges $\{y_{i}z_{j} : 1 \leq i \leq p,
\ 1 \leq j \leq q\}$. The following conditions are equivalent:

\begin{itemize}
\item $I$ is a positive instance of the \textbf{MONOTONE SAT} problem.

\item The induced subgraph of $G$ with vertices ${y_{1},y_{2},z_{1}}$ is generating.

\item The induced subgraph of $H$ with vertices ${y_{1},...,y_{p}%
,z_{1},...,z_{q}}$ is generating.
\end{itemize}
\end{proof}

\section{Bipartite Graphs Without Cycles of Length 6}

In this section we present efficient algorithms for recognizing generating
subgraphs, recognizing well-covered graphs, and finding the vector space
$WCW(G)$, in the restricted case that the input graph $G$ is bipartite without
cycles of length $6$.

\begin{lemma}
\label{indd2} Let $G$ be a bipartite graph without cycles of length $6$, and
let $B$ be a bipartite subgraph of $G$. Then $N_{2}(B)$ is independent.
\end{lemma}

\begin{proof}
Denote the vertex sets of bipartition of $B$ by $B_{X}$ and $B_{Y}$. Since $G$
is bipartite, $N(B_{X})\cap N(B_{Y})=\emptyset$. Therefore,
\[
N_{2}(B)=(N_{2}(B_{X})\cap N_{3}(B_{Y}))\cup(N_{3}(B_{X})\cap N_{2}(B_{Y}))
\]
All vertices of $N_{2}(B_{X})\cap N_{3}(B_{Y})$ belong to the same vertex set
of bipartition of $G$. Therefore, this set is independent. Similarly,
$N_{3}(B_{X})\cap N_{2}(B_{Y})$ is independent as well. Assume on the contrary
that there exist two adjacent vertices, $x^{\prime\prime}\in N_{2}(B_{X})\cap
N_{3}(B_{Y})$ and $y^{\prime\prime}\in N_{3}(B_{X})\cap N_{2}(B_{Y})$. Hence,
there exist vertices $x^{\prime}\in N(x^{\prime\prime})\cap N(B_{X})\cap
N_{2}(B_{Y})$ and $y^{\prime}\in N(y^{\prime\prime})\cap N_{2}(B_{X})\cap
N(B_{Y})$. Let $x\in N(x^{\prime})\cap B_{X}$ and $y\in N(y^{\prime})\cap
B_{Y}$. There exists a cycle of length $6$ in $G$, $(x^{\prime\prime
},x^{\prime},x,y,y^{\prime},y^{\prime\prime})$, which is a contradiction.
Consequently, $N_{2}(B)$ is independent.
\end{proof}

\begin{theorem}
\label{gsbc6p} There exists an $O(|V(G)|^{2})$ algorithm which solves the
following problem:\newline Input: A bipartite graph $G$ without cycles of
length $6$, and an induced subgraph $B$ of $G$.\newline Question: Is $B$ generating?
\end{theorem}

\begin{proof}
Let $B$ be an induced complete bipartite subgraph of $G$ on vertex sets of
bipartition $B_{X}$ and $B_{Y}$. Since $G$ is bipartite, $N\left(
B_{X}\right)  \cap N\left(  B_{Y}\right)  =\emptyset$. Define
\[
D_{1}=N(B_{X}\cup B_{Y})=(N(B_{X})\cap N_{2}(B_{Y}))\cup(N_{2}(B_{X})\cap
N(B_{Y})),
\]
and
\[
D_{2}=N_{2}(B_{X}\cup B_{Y})=(N_{2}(B_{X})\cap N_{3}(B_{Y}))\cup(N_{3}%
(B_{X})\cap N_{2}(B_{Y})).
\]
Clearly, $B$ is generating if and only if there exists an independent set in
$D_{2}$ which dominates $D_{1}$. However, by Lemma \ref{indd2}, $D_{2}$ is
independent. Hence, $B$ is generating if and only if $D_{2}$ dominates $D_{1}%
$. The following algorithm makes that decision.

\bigskip

\begin{algorithm}[H]
\DontPrintSemicolon
{ $D \longleftarrow \bigcup_{v \in B_{X} \cup B_{Y}} N(v)$  \; }
{ $D1 \longleftarrow D \setminus (B_{X} \cup B_{Y})$  \; }
{ \ForEach{$v \in D1$}
{ \If{$N(v) \setminus D = \emptyset$}
{ \Return $FALSE$\;}
}
}
{ \Return $TRUE$}
\caption{Generating(G,B)}
\label{generating}
\end{algorithm}

\bigskip

\textbf{Complexity of Algorithm \ref{generating}:} The graph is stored as a
boolean matrix. Sets of vertices are stored as boolean arrays of length
$|V(G)|$. Hence, deciding whether an element belongs to a set is done in
$O(1)$ time, while basic operations on sets such as union and difference are
implemented in $O(|V(G)|)$ time. Finding the set of neighbors of a vertex is
completed in $O(|V(G)|)$ time.

In Line 1 and Line 2 there are $O(|V(G)|)$ operations on sets, each takes
$O(|V(G)|)$ time. Therefore, the complexity of Line 1 and Line 2 is
$O(|V(G)|^{2})$. The \textbf{foreach} loop at Line 3 has $O(|V(G)|)$
iterations. Hence, Line 4 is performed $O(|V(G)|)$ times. One run of Line 4
takes $O(|V(G)|)$ time for evaluating the \textbf{if} condition. Therefore,
the total time Algorithm \ref{generating} spends in Line 4 is $O(|V(G)|^{2})$.
Thus the complexity of the whole algorithm is $O(|V(G)|^{2})$ as well.
\end{proof}

Note that also Theorem 8 of \cite{lt:wwc456}, deals with a restricted case in
which a subgraph $B$ is generating if and only if $N_{2}(B)$ dominates
$N(B_{X})\Delta N(B_{Y})$. However, in that restricted case $G$ does not
contain cycles of lengths $5$, $6$ and $7$. Algorithm \ref{generating} does
not fit \cite{lt:wwc456}, since in \cite{lt:wwc456} the set $N_{2}(B)$ is not
necessarily independent.

If $G$ is a graph without cycles of length $6$, and $B$ is a subgraph of $G$
on vertex sets of bipartition $B_{X}$ and $B_{Y}$, then $min(|B_{X}%
|,|B_{Y}|)\leq2$. In the remaining part of Section 4 the following notation is
used. (See Figure \ref{Fig2}.) The set of vertices $X$ is independent, and
$1\leq\left\vert X\right\vert \leq2$. If $\left\vert X\right\vert =1$, then
$X=\{x\}$, otherwise $X=\{x_{1},x_{2}\}$. Moreover, $Y=\{y_{1},...,y_{k}%
\}=\bigcap\limits_{x\in X}N(x)$. For every $1\leq i\leq k$ define
$A_{i}=N(y_{i})\setminus X$ and $Z_{i}=N_{2}(y_{i})\cap N_{3}(X)$.
More notation: $S=N(X)\setminus Y$ and $S^{\prime}=N(S)\setminus X$. Note that
if $\left\vert X\right\vert =1$ then $S=S^{\prime}=\emptyset$.

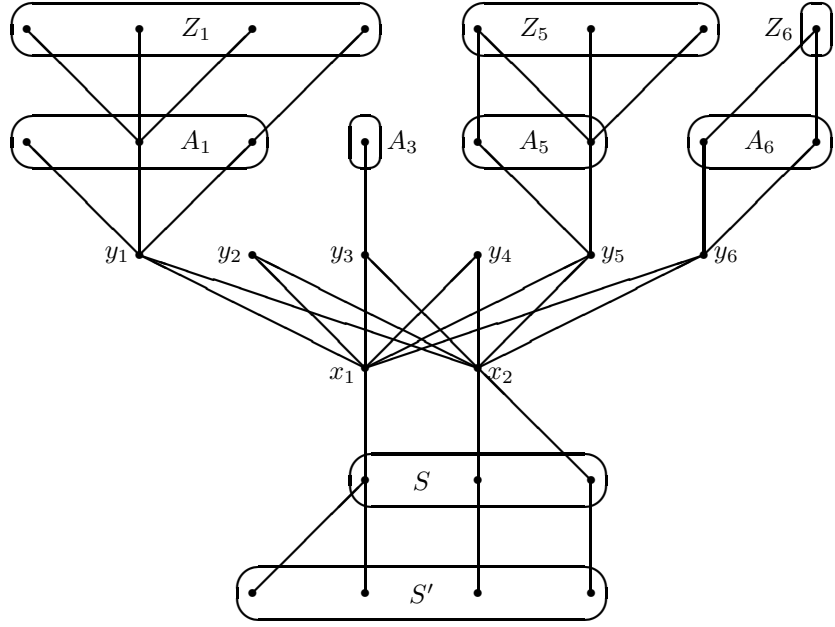
\begin{figure}[h]
\setlength{\unitlength}{1.0cm} \begin{picture}(20,9)\thicklines
\multiput(3,5.5)(1.5,0){6}{\circle*{0.1}}
\multiput(6,4)(1.5,0){2}{\circle*{0.1}}
\put(5.7,3.9){\makebox(0,0){$x_{1}$}}
\put(7.8,3.9){\makebox(0,0){$x_{2}$}}
\put(6,4){\line(-2,1){3}}
\put(6,4){\line(-1,1){1.5}}
\put(6,4){\line(0,1){1.5}}
\put(6,4){\line(1,1){1.5}}
\put(6,4){\line(2,1){3}}
\put(6,4){\line(3,1){4.5}}
\put(7.5,4){\line(-3,1){4.5}}
\put(7.5,4){\line(-2,1){3}}
\put(7.5,4){\line(-1,1){1.5}}
\put(7.5,4){\line(0,1){1.5}}
\put(7.5,4){\line(1,1){1.5}}
\put(7.5,4){\line(2,1){3}}
\put(2.7,5.5){\makebox(0,0){$y_{1}$}}
\put(4.2,5.5){\makebox(0,0){$y_{2}$}}
\put(5.7,5.5){\makebox(0,0){$y_{3}$}}
\put(7.8,5.5){\makebox(0,0){$y_{4}$}}
\put(9.3,5.5){\makebox(0,0){$y_{5}$}}
\put(10.8,5.5){\makebox(0,0){$y_{6}$}}
\multiput(1.5,7)(1.5,0){8}{\circle*{0.1}}
\put(3,5.5){\line(-1,1){1.5}}
\put(3,5.5){\line(0,1){1.5}}
\put(3,5.5){\line(1,1){1.5}}
\put(6,5.5){\line(0,1){1.5}}
\put(9,5.5){\line(-1,1){1.5}}
\put(9,5.5){\line(0,1){1.5}}
\put(10.5,5.5){\line(0,1){1.5}}
\put(10.5,5.5){\line(1,1){1.5}}
\multiput(1.5,8.5)(1.5,0){8}{\circle*{0.1}}
\put(3,7){\line(-1,1){1.5}}
\put(3,7){\line(0,1){1.5}}
\put(3,7){\line(1,1){1.5}}
\put(4.5,7){\line(1,1){1.5}}
\put(7.5,7){\line(0,1){1.5}}
\put(9,7){\line(-1,1){1.5}}
\put(9,7){\line(0,1){1.5}}
\put(9,7){\line(1,1){1.5}}
\put(10.5,7){\line(1,1){1.5}}
\put(12,7){\line(0,1){1.5}}
\put(3,7){\oval(3.4,0.7)}
\put(6,7){\oval(0.4,0.7)}
\put(8.25,7){\oval(1.9,0.7)}
\put(11.25,7){\oval(1.9,0.7)}
\put(3.75,8.5){\oval(4.9,0.7)}
\put(9,8.5){\oval(3.4,0.7)}
\put(12,8.5){\oval(0.4,0.7)}
\put(3.75,7){\makebox(0,0){$A_{1}$}}
\put(6.5,7){\makebox(0,0){$A_{3}$}}
\put(8.25,7){\makebox(0,0){$A_{5}$}}
\put(11.25,7){\makebox(0,0){$A_{6}$}}
\put(3.75,8.5){\makebox(0,0){$Z_{1}$}}
\put(8.25,8.5){\makebox(0,0){$Z_{5}$}}
\put(11.5,8.5){\makebox(0,0){$Z_{6}$}}
\multiput(6,2.5)(1.5,0){3}{\circle*{0.1}}
\put(6,2.5){\line(0,1){1.5}}
\put(7.5,2.5){\line(0,1){1.5}}
\put(9,2.5){\line(-1,1){1.5}}
\multiput(4.5,1)(1.5,0){4}{\circle*{0.1}}
\put(4.5,1){\line(1,1){1.5}}
\put(6,1){\line(0,1){1.5}}
\put(7.5,1){\line(0,1){1.5}}
\put(9,1){\line(0,1){1.5}}
\put(7.5,2.5){\oval(3.4,0.7)}
\put(6.75,1){\oval(4.9,0.7)}
\put(6.75,2.5){\makebox(0,0){$S$}}
\put(6.75,1){\makebox(0,0){$S'$}}
\end{picture}
\caption{Notation used in this Section. The sets $A_{2}$, $A_{4}$, $Z_{2}$,
$Z_{3}$, and $Z_{4}$ are empty. Since $S^{\prime}$ dominates $S$, and $Z_{i}$
dominates $A_{i}$ for each $i \in\{2,4,5,6\}$, the subgraph $G[\{x_{1}%
,x_{2},y_{2},y_{4},y_{5},y_{6}\}]$ is generating.}%
\label{Fig2}%
\end{figure}

\begin{lemma}
\label{gssbc6} Let $G$ be a bipartite graph without cycles of length $6$. Let
$X$ be an independent set of vertices such that $1\leq\left\vert X\right\vert
\leq2$, and let $Y=\{y_{1},...,y_{k}\}=\bigcap\limits_{x\in X}N(x)$. Assume
that $\left\vert X\right\vert \leq\left\vert Y\right\vert $ and $B=G[X\cup Y]$
is generating. Let $Y^{\prime}$ be a nonempty subset of $Y$, such that
$y_{i}\in Y\setminus Y^{\prime}$ implies $A_{i}\neq\emptyset$, for every
$1\leq i\leq k$. Then $B^{\prime}=G[X\cup Y^{\prime}]$ is generating.
\end{lemma}

\begin{proof}
By Lemma \ref{indd2}, $N_{2}(V(B))$ and $N_{2}(V(B^{\prime}))$ are independent
sets. It remains to prove that $N(V(B^{\prime}))$ is dominated by
$N_{2}(V(B^{\prime}))$.

Clearly, $N(V(B))=S\cup(\bigcup\limits_{1\leq i\leq k}A_{i})$ and
$N_{2}(V(B))=S^{\prime}\cup(\bigcup\limits_{1\leq i\leq k}Z_{i})$. Since $B$
is generating, $N_{2}(V(B))$ dominates $N(V(B))$. Therefore, $S^{\prime}$
dominates $S$, and $Z_{i}$ dominates $A_{i}$ for every $1\leq i\leq k$.

Let $I=\{1\leq i\leq k:y_{i}\in Y^{\prime}\}$ and $\overline{I}=\{1\leq i\leq
k:y_{i}\not \in Y^{\prime}\}$. Then $N(V(B^{\prime}))=S\cup(Y\setminus
Y^{\prime})\cup(\bigcup_{i\in I}A_{i})$ and $N_{2}(V(B^{\prime}))=S^{\prime
}\cup(\bigcup\limits_{i\in\overline{I}}A_{i})\cup(\bigcup\limits_{i\in I}%
Z_{i})$. For every $i\in\overline{I}$ it holds that $A_{i}\neq\emptyset$, and
therefore $A_{i}\subseteq N_{2}(V(B^{\prime}))$ dominates $y_{i}\in
N(V(B^{\prime}))$. For every $i\in I$ it holds that $Z_{i}\subseteq
N_{2}(V(B^{\prime}))$ dominates $A_{i}\subseteq N(V(B^{\prime}))$. Hence,
$N(V(B^{\prime}))$ is dominated by $N_{2}(V(B^{\prime}))$.
\end{proof}

\begin{lemma}
\label{maxgsbc6} The following problem can be solved in $O(|V(G)|^{3})$
time.\newline Input: A bipartite graph $G$ without cycles of length $6$, and a
vertex $x\in V(G)$. \newline Output: A maximal set $T\subseteq N(x)$ such that
$G[\{x\}\cup T]$ is generating, or $\emptyset$ if such a set does not exist.
\end{lemma}

\begin{proof}
Let $N(x)=Y=\{y_{1},...,y_{k}\}$. Let $I=\{1\leq i\leq k:A_{i}\subseteq
N(Z_{i})\}$ and $\overline{I}=\{1\leq i\leq k:A_{i}\setminus N(Z_{i}%
)\neq\emptyset\}$. Let $T=\{y_{i}:i\in I\}$. Note that if $i\in\overline{I}$,
then $A_{i}\neq\emptyset$. Each of the sets $\bigcup\limits_{i=1}^{k}A_{i}$
and $\bigcup\limits_{i=1}^{k}Z_{i}$ is independent, since all its vertices
belong to the same vertex set of bipartition of $G$.

Define $D_{1}=N(\{x\}\cup T)=(Y\setminus T)\cup(\bigcup\limits_{i\in I}A_{i})$
and $D_{2}=N_{2}(\{x\}\cup T)=(\bigcup\limits_{i\in I}Z_{i})\cup
(\bigcup\limits_{i\in\overline{I}}A_{i})$. By Lemma \ref{indd2}, $D_{2}$ is independent.

Assume $T\neq\emptyset$. We prove that $G[\{x\}\cup T]$ is generating. By
Lemma \ref{indd2}, it is enough to prove that $D_{1}\subseteq N(D_{2})$. Let
$1\leq i\leq k$. If $i\in\overline{I}$, then, by definition of $I$, $A_{i}%
\neq\emptyset$. Hence, $A_{i}\subseteq D_{2}$ dominates $y_{i}$. However, if
$i\in I$, then $A_{i}\subseteq D_{1}$. Therefore, by definition of $I$,
$A_{i}$ is dominated by $Z_{i}\subseteq D_{2}$.

On the other hand, let $B^{\prime}$ be an induced complete bipartite subgraph
of $G$ on vertex sets of bipartition $\{x\}$ and $T^{\prime}$, such that there
exists $y_{i} \in T^{\prime}\setminus T$. There exits $a_{i} \in A_{i}$ which
is not dominated by $Z_{i}$. Therefore, $N_{2}(\{x\} \cup T^{\prime})$ does
not dominate $N(\{x\} \cup T^{\prime})$, and $B^{\prime}$ is not generating.
Hence, $T$ is a maximal set such that $G[\{x\}\cup T]$ is generating.

Let us conclude the proof by presenting the algorithm.

\bigskip

\begin{algorithm}[H]
\DontPrintSemicolon
{ $T \longleftarrow \emptyset$  \; }
{ \ForEach{$y \in N(x)$}
{ { $A \longleftarrow N(y) \setminus \{x\}$ \; }
{ $flag \longleftarrow TRUE$ \;}
{ \ForEach{$a \in A$ }
{
{ \If{$N(a) \setminus N(x) = \emptyset$}
{
{ $flag \longleftarrow FALSE$\;}
{\bf break }
}
}
}
}
{ \If{$flag$}
{ $T \longleftarrow T \cup \{y\}$ }
}
}
}
{ \Return $T$}
\caption{$MaxGen1(G,x)$\label{maxgen}}
\end{algorithm}

\bigskip

\textbf{Complexity of Algorithm \ref{maxgen}:} Each of the \textbf{foreach}
loops in Lines 2 and 5 has $O(|V(G)|)$ iterations. Therefore, Lines 6-8 are
performed $O(|V(G)|^{2})$ times. The condition in Line 6 involves operations
on sets. Hence, it takes $O(|V(G)|)$ time to evaluate it once. The total time
Algorithm \ref{maxgen} spends in Line 6 is $O(|V(G)|^{3})$. This is also the
complexity of the algorithm.
\end{proof}

\begin{lemma}
\label{w0gsbc6} Let $G$ be a bipartite graph without cycles of length 6, and
let $B$ be a generating subgraph of $G$ on vertex sets of bipartition $X$ and
$Y = \{ y_{1},...,y_{k}\}$, such that $k \geq2$. Let $w \in WCW(G)$. Let $1
\leq j \leq k$ such that $A_{j} \neq\emptyset$. Then $w(y_{j})=0$.
\end{lemma}

\begin{proof}
The fact that $B$ is generating implies that $w(X) = w(Y)$.

Let $Y^{\prime}= Y \setminus\{y_{j}\}$, and let $B^{\prime}=G[X \cup
Y^{\prime}]$. By Lemma \ref{gssbc6}, $B^{\prime}$ is generating. Therefore,
$w(X) = w(Y^{\prime})$. Consequently, $w(Y-Y^{\prime}) = w(y_{j})=0$.
\end{proof}

\begin{lemma}
\label{max2gsbc6} The following problem can be solved in $O(|V(G)|^{2})$
time.\newline Input: A bipartite graph $G$ without cycles of length $6$, and
two vertices $x_{1},x_{2}\in V(G)$. \newline Output: A maximal set $T\subseteq
N(x_{1})\cap N(x_{2})$ such that $\left\vert T\right\vert \geq2$ and
$G[\{x_{1},x_{2}\}\cup T]$ is generating, or $\emptyset$ if such a set does
not exist.
\end{lemma}

\begin{proof}
Let $X=\{x_{1},x_{2}\}$ and $Y=N(x_{1})\cap N(x_{2})=\{y_{1},...,y_{k}\}$.
Suppose $k\geq2$. Let $S=N(X)\setminus Y$, and $S^{\prime}=N(S)\setminus X$.
(See Fig. \ref{Fig2}.) If $S$ is not dominated by $S^{\prime}$, then such a
set $T$ does not exist, and the algorithm outputs $\emptyset$. Hence, assume
$S\subseteq N(S^{\prime})$.

Assume on the contrary that $(\bigcup\limits_{1\leq i\leq k}A_{i})\cup S$ is
not independent. There exist two adjacent vertices, $a\in\bigcup\limits_{1\leq
i\leq k}A_{i}$ and $s\in S$. Assume without loss of generality that $s\in
N(x_{1})$. There exists $y\in Y\cap N(a)$, and $y^{\prime}\in Y\setminus
\{y\}$. Hence, $(a,y,x_{2},y^{\prime},x_{1},s)$ is a cycle of length $6$,
which is a contradiction. Therefore, $(\bigcup\limits_{1\leq i\leq k}%
A_{i})\cup S$ is independent. By Lemma \ref{indd2}, the set $S^{\prime}%
\cup(\bigcup\limits_{i=1}^{k}Z_{i})$ is independent as well.

Let $T=\{y_{i}\in Y:A_{i}\subseteq N(Z_{i})\}$. Note that if $y_{i}\not \in
T$, then $A_{i}\neq\emptyset$. Define%

\[
D_{1} = N( X \cup T ) = \{y_{i} : y_{i} \not \in T\} \cup\{A_{i} : y_{i} \in
T\} \cup S
\]
and%

\[
D_{2} = N_{2}( X \cup T ) = \{Z_{i} : y_{i} \in T\} \cup\{A_{i} : y_{i}
\not \in T\} \cup S^{\prime}%
\]
By Lemma \ref{indd2}, $D_{2}$ is independent.

Suppose $\left\vert T\right\vert \geq2$. We prove that $G[X\cup T]$ is
generating. By Lemma \ref{indd2}, it is enough to prove that $D_{1}\subseteq
N(D_{2})$. We assume that $S\subseteq D_{1}$ is dominated by $S^{\prime
}\subseteq N(D_{2})$. Let $1\leq i\leq k$. If $y_{i}\in D_{1}$, then, by
definition of $T$, $A_{i}\neq\emptyset$. Hence, $A_{i}\subseteq D_{2}$
dominates $y_{i}$. If $y_{i}\not \in D_{1}$, then $A_{i}\subseteq D_{1}$.
Therefore, by definition of $T$, $A_{i}$ is dominated by $Z_{i}\subseteq
D_{2}$.

On the other hand, let $B^{\prime}$ be an induced complete bipartite subgraph
of $G$ on vertex sets of bipartition $X$ and $T^{\prime}$, such that there
exists $y_{i} \in T^{\prime}\setminus T$. There exits $a_{i} \in A_{i}$ which
is not dominated by $Z_{i}$. Therefore, $N_{2}(X \cup T^{\prime})$ does not
dominate $N(X \cup T^{\prime})$, and $B^{\prime}$ is not generating.

Let us conclude the proof by presenting the algorithm.

\bigskip

\begin{algorithm}[H]
\DontPrintSemicolon
{ $T \longleftarrow \emptyset$  \; }
{ $X \longleftarrow \{ x_{1}, x_{2} \}$  \; }
{ $Y \longleftarrow N(x_{1}) \cap N(x_{2})$  \; }
{ \If{$|Y| < 2$}
{ \Return $\emptyset$}
}
{ \ForEach{$s \in N(x_{1}) \Delta N(x_{2})$}
{
{ \If{$N(s) \setminus X = \emptyset$}
{ \Return $\emptyset$}
}
}
}
{ \ForEach{$y \in Y$}
{
{ $flag \longleftarrow TRUE$ \;}
{\ForEach{$a \in N(y) \setminus X$}
{ \If{$N(a) \setminus Y = \emptyset$}
{
{ $flag \longleftarrow FALSE$ \;}
{\bf break }
}
}
{ \If{$flag = TRUE$}
{ $T \longleftarrow T \cup \{y\}$ }
}
}
}
}
{ \If{$|T| < 2$}
{ \Return $\emptyset$}
}
{ \Return $T$}
\caption{$MaxGen2(G,x_{1},x_{2})$}
\label{maxgen2}
\end{algorithm}

\bigskip

\textbf{Complexity of Algorithm \ref{maxgen2}:} The complexity of Lines 1-5 is
$O(|V(G)|)$. The \textbf{foreach} loop in Line 6 has $O(|V(G)|)$ iterations.
In each iteration the condition of Line 7 is evaluated in $O(|V(G)|)$ time.
Therefore, the complexity of Lines 6-8 is $O(|V(G)|^{2})$.

If a vertex $a\in\bigcup\limits_{1\leq i\leq k}A_{i}$ was adjacent to 3
distinct vertices, $y$, $y^{\prime}$ and $y^{\prime\prime}$ of $Y$, then $G$
contained a cycle of length 6, $(y,a,y^{\prime},x_{1},y^{\prime\prime},x_{2}%
)$. Therefore, every $a\in\bigcup\limits_{1\leq i\leq k}A_{i}$ is adjacent to
two vertices of $Y$ at most. Hence, the total number of iterations of the
\textbf{foreach} loop of Line 11 is $O(|V(G)|)$. Each evaluation of the
condition of Line 12 takes $O(|V(G)|)$ time, and the algorithm spends
$O(|V(G)|^{2})$ time in Line 12. Hence, the complexity of Lines 9-16 is
$O(|V(G)|^{2})$. The complexity of Lines 17-19 is $O(|V(G)|)$. The total
complexity of Algorithm \ref{maxgen2} is $O(|V(G)|^{2})$.
\end{proof}

\begin{theorem}
\label{wcwbc6} The following problem can be solved in $O(|V(G)|^{4})$
time.\newline Input: A bipartite graph $G$ without cycles of length $6$.
\newline Output: $WCW(G)$.
\end{theorem}

\begin{proof}
The following algorithm receives as its input a bipartite graph $G$ without
cycles of length 6. The algorithm outputs a list of restrictions. For every
function $w:V\left(  G\right)  \longrightarrow\mathbb{R}$ it holds that $w\in
WCW(G)$ if and only if $w$ satisfies all restrictions outputted by the
algorithm. The algorithm invokes the $MaxGen1$ algorithm defined in the proof
of Lemma \ref{maxgsbc6}, and the $MaxGen2$ algorithm defined in the proof of
Lemma \ref{max2gsbc6}.

\bigskip

\begin{algorithm}[H]
\DontPrintSemicolon
{ \ForEach{$v \in V(G)$}
{
{ $T \longleftarrow MaxGen1(G,v)$ \; }
{ \If{$T \neq \emptyset$}
{ {\bf output} $w(T)=w(v)$ \; }
}
{ \If{$|T| \geq 2$}
{ \ForEach{$t \in T$}
{ \If{$N(t) \neq \{ v \}$}
{ {\bf output} $w(t)=0$}
}
}
}
}
}
{ \ForEach{$v_{1} \in V(G)$}
{ \ForEach{$v_{2} \in V(G) \setminus \{v_{1}\}$}
{
{ $T \longleftarrow MaxGen2(G,v_{1},v_{2})$ \; }
{ \If{$|T| \geq 2$}
{
{ {\bf output} $w(T)=w(\{v_{1},v_{2}\})$ \; }
{ \ForEach{$t \in T$}
{ \If{$N(t) \neq \{ v_{1},v_{2} \}$}
{ {\bf output} $w(t)=0$}
}
}
}
}
}
}
}
\caption{WCW(G)\label{wcw}}
\end{algorithm}

\bigskip

\textbf{Complexity of Algorithm \ref{wcw}:} The complexity of $MaxGen1$ is
$O(|V(G)|^{3})$. That routine is invoked $O(|V(G)|)$ times by this algorithm.
The complexity of Lines 1-8 is $O(|V(G)|^{4})$. The complexity of $MaxGen2$ is
$O(|V(G)|^{2})$. That routine is invoked $O(|V(G)|^{2})$ times by this
algorithm. The complexity of Lines 9-16 is $O(|V(G)|^{4})$.
\end{proof}

\begin{corollary}
\label{wcbc6} The following problem can be solved in $O(|V(G)|^{4})$
time.\newline Input: A bipartite graph $G$ without cycles of length $6$.
\newline Question: Is $G$ well-covered?
\end{corollary}

\begin{proof}
In order to decide whether a graph $G$ is well-covered, one can find the
vector space $WCW(G)$, and decide whether it includes the uniform function $w
\equiv1$. However, we present an algorithm which is faster than Algorithm
\ref{wcw}, although it has the same computational complexity.

\begin{algorithm}[H]
\DontPrintSemicolon
{ \ForEach{$v \in V(G)$}
{
{ $T \longleftarrow MaxGen1(G,v)$ \; }
{ \If{$|T| > 1$}
{ \Return $FALSE$}
}
}
}
{ \ForEach{$v_{1} \in V(G)$}
{ \ForEach{$v_{2} \in V(G) \setminus \{v_{1}\}$}
{
{ $T \longleftarrow MaxGen2(G,v_{1},v_{2})$ \; }
{ \If{$|T| > 2$}
{ \Return $FALSE$}
}
{ \If{$|T| = 2$}
{ \ForEach{$t \in T$}
{ \If{$N(t) \neq \{ v_{1},v_{2} \}$}
{ \Return $FALSE$}
}
}
}
}
}
}
{ \Return $TRUE$}
\caption{WC(G)\label{wc}}
\end{algorithm}

\bigskip

\textbf{Complexity of Algorithm \ref{wc}:} The complexity of $MaxGen1$ is
$O(|V(G)|^{3})$, and it is invoked $O(|V(G)|)$ times by Algorithm \ref{wc}.
The complexity of $MaxGen2$ is $O(|V(G)|^{2})$, and it is invoked
$O(|V(G)|^{2})$ times by Algorithm \ref{wc}. The total complexity of Algorithm
\ref{wc} is $O(|V(G)|^{4})$.
\end{proof}

\section{Graphs Without Cycles of Lengths 3, 4, 5, 7}

In this section, $G$ is a graph without cycles of lengths $3$, $4$, $5$ and
$7$. Since $G$ does not contain small odd cycles, $G[N_{3}(v)]$ is bipartite
for every $v\in V(G)$. Define $L(G)=\{v\in V(G)\ |$ $d(v)=1\}$, and
$S_{x}=N(x)\setminus N(L(G))$ for every $x\in V(G)\setminus L(G)$. For every
$y\in S_{x}$ it holds that $N(y)\cap N_{2}(x)\cap L(G)=\emptyset$. Therefore,
$N(y)\cap N_{2}(x)\subseteq N(N_{2}(y)\cap N_{3}(x))$.

The main result of this section is a polynomial characterization of the vector
space $WCW(G)$ for graphs without cycles of lengths 3, 4, 5 and 7.

\begin{lemma}
\label{maxsx} Let $G$ be a graph without cycles of lengths $3$, $4$, $5$ and
$7$, and let $x\in V(G)\setminus L(G)$. Then $S_{x}$ is a maximal set with the
following two properties:

\begin{enumerate}
\item $S_{x} \subseteq N(x)$

\item If $S_{x}\neq\emptyset$ then $G[\{x\}\cup S_{x}]$ is generating.
\end{enumerate}
\end{lemma}

\begin{proof}
Obviously, $S_{x}\subseteq N(x)$. We assume $S_{x}\neq\emptyset$, and prove
that $G[\{x\}\cup S_{x}]$ is generating. Let $T_{1}=N_{2}(S_{x})\cap N_{3}%
(x)$. Clearly, $T_{1}$ dominates $N(S_{x})\cap N_{2}(x)$. Define $T_{2}%
=N_{2}(x)\cap L(G)$. It holds that $T_{2}$ dominates $N(x)\setminus(S_{x}\cup
L(G))$.

Since $T_{1} \subseteq N_{3}(x)$, and $T_{2} \subseteq N_{2}(x) \cap L(G)$,
the set $T_{1} \cup T_{2}$ is independent. Let $T$ be any maximal independent
set of $G \setminus N[\{x\} \cup S_{x}]$, which contains $T_{1} \cup T_{2}$.
Obviously, $T$ is a witness that $G[\{x\} \cup S_{x}]$ is generating.

It remains to prove the maximality of $S_{x}$. Let $S_{x}^{\prime}\subseteq
N(x)$ such that there exists $y \in S_{x}^{\prime}\setminus S_{x}$. We prove
that $G[\{x\}\cup S_{x}^{\prime}]$ is not generating. There exists $l \in N(y)
\cap L(G)\cap N_{2}(x)$. Any independent set $S\subseteq V(G) \setminus
N[\{x\} \cup S_{x}^{\prime}]$ is not a witness that $G[\{x\}\cup S_{x}%
^{\prime}]$ is generating, because it does not dominate $l$.
\end{proof}

\begin{lemma}
\label{wnl} Let $G$ be a graph without cycles of lengths 3, 4, 5 and 7. Then
$w(x)=w(L(G)\cap N(x))$ for every $x \in N(L(G))$ and for every $w\in WCW(G)$.
\end{lemma}

\begin{proof}
Let $x \in N(L(G))$ and let $T$ be a maximal independent set of $G \setminus
N[x]$, which contains $N_{2}(x)$. Then $T_{1} = T \cup\{x\}$ and $T_{2} = T
\cup( N(x) \cap L(G) )$ are maximal independent sets of $G$. The fact that
$w(T_{1}) = w(T_{2})$ implies $w(x)= w(L(G) \cap N(x))$.
\end{proof}

\begin{lemma}
\label{weq0} Let $G$ be a graph without cycles of lengths $3$, $4$, $5$ and
$7$. Then $w(y)=0$ for every vertex $y\in V(G)\setminus N[L(G)]$ and for every
$w\in WCW(G)$.
\end{lemma}

\begin{proof}
Let $x_{1}$ and $x_{2}$ be two distinct neighbors of $y$. Since $y
\not \in N[L(G)]$, it holds that $y \in S_{x_{1}}$ and $y \in S_{x_{2}}$. The
fact that $G$ does not contain cycles of length 4 implies that $\{y\} =
S_{x_{1}} \cap S_{x_{2}}$. By Lemma \ref{maxsx}, each of $G[\{x_{1}\} \cup
S_{x_{1}}]$ and $G[\{x_{2}\} \cup S_{x_{2}}]$ is generating. Therefore,
$w(x_{1})=w(S_{x_{1}})$ and $w(x_{2})=w(S_{x_{2}})$, for every $w \in WCW(G)$.

Let $B = G[\{x_{1},x_{2}\} \cup S_{x_{1}} \cup S_{x_{2}}]$. Clearly, $B$ is an
induced bipartite subgraph of $G$, but it is not necessarily complete.
Therefore, it does not fit the definition of a generating subgraph. However,
we prove that there exists a set, $T^{*}$, such that $T^{*} \cup\{x_{1},
x_{2}\}$ and $T^{*} \cup S_{x_{1}} \cup S_{x_{2}}$ are maximal independent
sets of $G$.

Let $T = ( N_{2}(S_{x_{1}} \cup S_{x_{2}}) \cap N_{3}(\{x_{1},x_{2}\}) ) \cup(
N_{2}(\{x_{1},x_{2}\}) \cap L(G) )$. It it easy to see that $T$ is independent
and dominates $N(B)$. Let $T^{*}$ be any maximal independent set of $G
\setminus N[B]$ which contains $T$. The fact that $T^{*} \cup\{x_{1}, x_{2}\}$
and $T^{*} \cup S_{x_{1}} \cup S_{x_{2}}$ are maximal independent sets of $G$
implies that $w( \{x_{1}, x_{2}\} ) = w( S_{x_{1}} \cup S_{x_{2}} )$.
Equivalently, $w(x_{1}) + w(x_{2}) = w(S_{x_{1}}) + w(S_{x_{2}}) - w(S_{x_{1}}
\cap S_{x_{2}})$. Therefore, $w(y) = 0$.
\end{proof}

\begin{theorem}
\label{wcwbg6} Let $G$ be a connected graph without cycles of lengths 3, 4, 5
and 7, which is not isomorphic to $K_{2}$. Let $w : V(G) \longrightarrow
\mathbb{R}$. Then the following conditions are equivalent.

\begin{enumerate}
\item $w \in WCW(G)$

\item $w(x)=w(N(x) \cap L(G))$ for every vertex $x \in V(G) \setminus L(G)$.
\end{enumerate}
\end{theorem}

\begin{proof}
If $G$ is an isolated vertex then $V(G)\setminus L(G)=\emptyset$, and every
function defined on $V(G)$ belongs to $WCW(G)$. Therefore, Condition 1 and
Condition 2 hold. In the remaining of the proof we assume that $\left\vert
V(G)\right\vert \geq3$.

Suppose that the first condition holds. By Lemmas \ref{wnl} and \ref{weq0},
the second condition holds as well.

Assume the second condition holds. Denote $N(L(G))=\{v_{1},...,v_{k}\}$. We
prove that the weight of every maximal independent set is $\sum\limits_{1\leq
i\leq k}w(v_{i})$. For every $1\leq i\leq k$ define $T_{i}=\{v_{i}%
\}\cup(N(v_{i})\cap L(G))$. Let $S$ be a maximal independent set of $G$. For
every $1\leq i\leq k$ either $S\cap T_{i}=\{v_{i}\}$ or $S\cap T_{i}%
=N(v_{i})\cap L(G)$. In both cases $w(S\cap T_{i})=w(v_{i})$. Note that
$T_{i}\cap T_{j}=\emptyset$ for every $1\leq i<j\leq k$. Moreover, $w(v)=0$
for every vertex $v\in V(G)\setminus(\bigcup\limits_{1\leq i\leq k}T_{i})$.
Hence,
\begin{align*}
w(S)  &  =\sum_{1\leq i\leq k}w(S\cap T_{i})+w(S\cap(V(G)\setminus
(\bigcup_{1\leq i\leq k}T_{i}))\\
&  =\sum_{1\leq i\leq k}w(v_{i})+0=w(N(L(G))).
\end{align*}
Therefore, the first condition holds.
\end{proof}

Bipartite graphs with girth at least $6$ are a restricted case of graphs
without cycles of lengths $3$, $4$, $5$, and $7$. Hence, we obtain the following.

\begin{corollary}
\label{wcwbg6c} Let $G$ be a connected bipartite graph with girth at least 6.
Assume that $G$ is not isomorphic to $K_{2}$. Let $w : V(G) \longrightarrow
\mathbb{R}$. Then the following conditions are equivalent.

\begin{enumerate}
\item $w \in WCW(G)$

\item $w(x)=w(N(x) \cap L(G))$ for every vertex $x \in V(G) \setminus L(G)$.
\end{enumerate}
\end{corollary}

\section{Conclusions and Future Work}

In Subsection 3.2 we considered graphs without cycles of lengths $3$ and $5$.
We proved that recognizing generating subgraphs isomorphic to $K_{1,2}$ is an
\textbf{NP-}complete task. By performing minor changes in the proof we showed
that for every $p\geq1$ and $q\geq2$, recognizing generating subgraphs
isomorphic to $K_{p,q}$ is also \textbf{NP-}complete. Hence, we conjecture the following.

\begin{conjecture}
\label{genmonotone} Let $i \leq p$ and $j \leq q$. Let $\Psi$ be a family of
graphs for which recognizing generating subgraphs isomorphic to $K_{i,j}$ is
\textbf{NP-}complete. Then recognizing generating subgraphs of $\Psi$
isomorphic to $K_{p,q}$ is \textbf{NP-}complete as well.
\end{conjecture}

We considered graphs which do not contain cycles of lengths $3$ and $5$. For
this family of graphs we proved that recognizing generating subgraphs is an
\textbf{NP-}complete problem. However, we still do not know the complexity
statuses of recognizing relating edges, deciding whether a graph is
well-covered, and finding the vector space $WCW$.

\end{document}